\documentclass[12pt,twoside]{article}
\usepackage[mathscr]{eucal}
\usepackage{amsmath,amsfonts,amssymb,amsthm,ulem}
\usepackage{hyperref}
\usepackage{times}

\usepackage{multibox}

\newcommand{\ul}[1]{{\underline{#1}}}

\def\be{\begin{equation}}
\def\ee{\end{equation}}
\def\ba{\begin{array}}
\def\ea{\end{array}}

\def\dps{\displaystyle}

\def\sd{S^\dagger}
\def\bsd{\bar S}

\def\BGST{Barnich:2004cr}
\def\BGadS{Barnich:2006pc}
\def\BGL{Batalin:2001je}
\def\GL{Grigoriev:2000rn}

\def\AGT{Alkalaev:2008gi}

\advance\textheight\topskip
\renewcommand{\tilde}{\widetilde}
\renewcommand{\hat}{\widehat}
\newtheorem{prop}{Proposition}[section]
\newtheorem{lemma}[prop]{Lemma}
\newtheorem{definition}[prop]{Definition}

\newcommand{\bref}[1]{\textbf{\ref{#1}}}

\newcommand{\im}{\mathop{\mathrm{Im}}}

\newcommand{\gh}[1]{\mathrm{gh}(#1)}

\newcommand{\module}[1]{\mathscr{#1}}

\newcommand{\modM}{\module{M}} 

\newcommand{\smalln}{n}
\newcommand{\smallq}{q}
\newcommand{\smalls}{s}

\newcommand{\dd}{\partial}
\renewcommand{\d}{\partial}

\renewcommand{\geq}{\,{\geqslant}\,}
\renewcommand{\leq}{\,{\leqslant}\,}

\newcommand{\binner}[2]{%
  {\langle}\kern-4.15pt{\langle}#1{,}\,#2{\rangle}\kern-4.15pt{\rangle}}
\newcommand{\commut}[2]{[#1{,}\,#2]}

\newcommand{\half}{\mathchoice{%
    \ffrac{1}{2}}{\frac{1}{2}}{\frac{1}{2}}{\frac{1}{2}}}

\newcommand{\ffrac}[2]{\raisebox{.5pt}%
  {\footnotesize$\displaystyle\frac{#1}{#2}$}\kern1pt}

\newcommand{\brst}{\mathsf{\Omega}}

\newcommand{\dl}[1]{\mathchoice{\ffrac{\dd}{\dd #1}}{\frac{\dd}{\dd
      #1}}{\ffrac{\dd}{\dd #1}}{\ffrac{\dd}{\dd #1}}}

\newcommand{\ddl}[2]{\ffrac{\dd #1}{\dd #2}}

\newcommand{\fR}{\mathbb{R}}

\newcommand{\derham}{\boldsymbol{d}}

\newcommand{\manifold}[1]{\mathscr{#1}}
\newcommand{\manX}{\manifold{X}}

\def\cI{\mathcal{I}}
\def\cJ{\mathcal{J}}
\def\cK{\mathcal{K}}

\def\cN{\mathcal{N}}

\def\cP{\mathcal{P}}

\def\cS{\mathcal{S}}
\def\cT{\mathcal{T}}

\def\cV{\mathcal{V}}

\numberwithin{equation}{section} \makeatletter
\@addtoreset{equation}{section}

\begin{document}

\begin{flushright}
\end{flushright}

~~
\vspace{0.5cm}
\begin{center}

{\Large\textbf{
Unified BRST description
of AdS gauge fields
}}

\vspace{.9cm}

{\large Konstantin Alkalaev and  Maxim Grigoriev}

\vspace{0.5cm}

\textit{I.E. Tamm Department of Theoretical Physics, \\P.N. Lebedev Physical
Institute,\\ Leninsky ave. 53, 119991 Moscow, Russia}

\vspace{1.5cm}

\begin{abstract}
A concise formulation for mixed-symmetry gauge fields
on AdS space is proposed. It is explicitly local, gauge invariant, and has manifest
AdS symmetry. Various other known formulations (including the original formulation of Metsaev and the unfolded formulation) can be derived  through the appropriate reductions and gauge fixing. As a byproduct, we also identify some new useful formulations of the theory that can be interesting for further developments. The formulation is presented in the BRST terms and  extensively uses Howe duality. In particular, the BRST operator
is a sum of the term associated to the spacetime isometry algebra
and the term associated to the Howe dual symplectic algebra.
\end{abstract}

\end{center}

\newpage
{
\small
\tableofcontents
}
\section{Introduction}

There have been numerous approaches to mixed-symmetry higher spin
gauge fields on the AdS space. In contrast to the totally
symmetric case where a simple Lagrangian formulation is
available~\cite{Fronsdal:1979vb}, describing mixed-symmetry AdS
fields is not so straightforward. In particular, general AdS gauge
fields have been described  much later~\cite{Metsaev:1995re} and
only at the level of equations of motion. Moreover, these
equations are not truly gauge-invariant as the gauge parameters
satisfy differential constraints. The true gauge fields were then
identified in \cite{Alkalaev:2003qv,Skvortsov:2009zu} within the
unfolded approach~\cite{Shaynkman:2000ts,Vasiliev:2001wa}. The
unfolded formulation of AdS gauge fields was recently proposed in
\cite{Boulanger:2008up,Boulanger:2008kw}. However, beyond the
totally symmetric field case \cite{Vasiliev:2001wa} this
formulation happens to be rather involved technically  because the
constraints imposed on the fields bring the respective projectors
to the equations of motion. As far as particular cases of
mixed-symmetry AdS fields are concerned there are other successful
approaches available in the literature \cite{Brink:2000ag,
Buchbinder:2001bs,Buchbinder:2006ge,Buchbinder:2007vq,Fotopoulos:2008ka,
Buchbinder:2008kw, Fotopoulos:2006ci, Zinoviev:2008ve,
Zinoviev:2009gh,
Reshetnyak:2008sf,Hallowell:2005np,Bastianelli:2009eh,Barnich:2005bn}.
Light-cone formulation for mixed-symmetry fields of any spins was
elaborated in \cite{Metsaev:1999ui, Metsaev:2000ja,
Metsaev:2004ee}.

Although all these formulations are believed (and partially
proved) to describe the same physical degrees of freedom their
explicit interrelations remain unclear. Moreover, further
developments and especially a search for mixed-symmetry fields
consistent interactions call for a simple and algebraically
transparent formulation that is free of the above difficulties.
This paper is devoted to constructing a candidate formulation that
meets these criteria. This turns out to be a natural
generalization of the recent~\cite{\AGT} (see also~\cite{\BGST}
for the case of Fronsdal fields) formulation for Minkowski space
mixed-symmetry fields. At the same time, it naturally generalizes
the formulation~\cite{\BGadS} of totally symmetric AdS fields to
the mixed-symmetry case. In particular, the equations of motion
and gauge symmetries has one and the same structure for massless
fields of arbitrary symmetry type in both Minkowski and AdS
spaces.

An important technical ingredient used throughout the paper is the
twisted version of the Howe dual \cite{Howe1} realisation of
symplectic and orthogonal algebras (in the case of Fronsdal
fields, \textit{i.e.} for $sp(4)$ algebra, this realization was
first used in~\cite{\BGadS}). Though equivalent to the usual one
in the space of polynomials it turns out inequivalent in the space
of formal power series because the equivalence transformation is
not well-defined in this space. In the same way as usual Howe
duality is useful in describing finite-dimensional irreducible
modules (\textit{e.g.} irreducibility conditions for one algebra
are highest weight conditions for its Howe dual) the twisted
realization also describes infinite-dimensional indecomposable
representations. This is crucial because both type of modules are
necessary to describe gauge fields. Namely, the generalized
curvatures take values in the indecomposable module (known as Weyl
module) while the generalized gauge potentials in the irreducible
modules (known as gauge modules) of the AdS isometry
algebra~\cite{Bekaert:2005vh}. The twisted Howe duality allows to
embed both type of modules in one and the same $o(d-1,2)-sp(2n)$
bimodule.

An attractive feature of the proposed construction is that the
irreducibility constraints commute with the equations of motion.
Strictly speaking, they are BRST invariant with respect to the
BRST operator defining the equations of motion and gauge
symmetries. This allows to simultaneously describe a collection of
irreducible fields such that an individual field can be then
singled out by the appropriate  constraints. This feature is
important from the string theory perspective, where the string
spectrum contains a huge collection of mixed-symmetry fields.
Although string theory leads to massive mixed-symmetry fields and
is not well-defined on AdS space, in the appropriate limit it is
expected to incorporate massless fields and to admit AdS
background (see e.g.~\cite{Bonelli:2003zu,Lindstrom:2003mg}).
Motivated by this relationship we also propose other equivalent
reformulations of the AdS mixed- symmetry fields including that
defined in terms of the ambient space and  based on the BRST
operator,\footnote{It is also similar to the formulation of \cite{
Fotopoulos:2006ci} for totally symmetric
 fields.} analogous to the standard one associated to the bosonic string .

\section{Preliminaries}
\label{sec:Preliminaries}

\subsection{Howe dual realizations}
\label{sec:howe}
In this section we introduce main technical tools of our construction that make
the whole consideration manifestly $o(d-1,2)$ covariant.

The anti-de Sitter spacetime AdS  can be described as a hyperboloid  $\manX$ embedded in the ambient flat pseudo-Euclidean space $\fR^{d+1}$. Labelling the coordinates in  $\fR^{d+1}$
as $X^A$, $A = 0,..., d$, the embedding equation is
\be
\label{hyper}
\eta_{AB}X^A X^B + 1 = 0\;,
\qquad \eta_{AB} = (- + \cdots + -)\;.
\ee
Infinitesimal   isometries
of the hyperboloid form a pseudo-orthogonal algebra $o(d-1,2)$.

Let $A^A_I$, where $A=0, ..., d$ and $I=0, ...., n-1$ be commuting variables transforming as
vectors of $o(d-1,2)$. The realization of $o(d-1,2)$ on the space  of
functions in $A^A_I$ reads
\be
J^{AB} =A^A_I\dl{A_B{}_I}-A^B_I\dl{A_A{}_I}\;.
\ee
The realization of $sp(2n)$ reads
\begin{equation}
\label{SPgenerators}
 T_{IJ}=A_I^A A_{JA}\,,
 \qquad
 T_I{}^J=\frac{1}{2}\,\{A^A_I, \dl{A_J^A}\}\,,
 \qquad
 T^{IJ}=\dl{A_I^A}\dl{A_{JA}}\,.
\end{equation}
These two algebras form a Howe dual pair $o(d-1,2)-sp(2n)$ \cite{Howe1}. The diagonal elements $T_I{}^I$ form a
basis in the Cartan subalgebra while $T^{IJ}$ and $T_I{}^J,\; I<J$ are the basis
elements of the appropriately chosen upper-triangular subalgebra. Let us note that $gl(n)$ algebra is realized by the generators $T_I{}^J$ as a subalgebra of $sp(2n)$ while
its $sl(n)$ subalgebra is generated by $T_I{}^J$ with $I\neq J$.

In what follows we also need to pick up a distinguished direction in the space of oscillators
$A_I^A$. Without loss of generality we take it along $A_0^A$ so that from now on we consider
variables $A^A_0$ and $A^A_i$, $i=1,..., n-1$ separately.
In particular, we identify $sp(2n-2)\subset sp(2n)$ subalgebra preserving the direction.
We use the following notation for some of $sp(2n-2)$ generators
\be
\label{glnot}
N_i{}^j\equiv T_i{}^j=A_i^A\dl{A_j^A}\,\,\,\;\;\; i\neq j\,,
\qquad
N_i=N_i{}^i\equiv T_i{}^i-\frac{d+1}{2}=A_i^A\dl{A_i^A}\,,
\ee
which form $gl(n-1)$ subalgebra, and
\be
\label{smallTr}
T_{ij} = A_i^A A_j{}_A\;,
\qquad
T^{ij} = \dl{A_i^A}\dl{A_j{}_A}\;,
\ee
that complete above set of elements to
$sp(2n-2)$ algebra.

In what follows we use two different realizations of $sp(2n)$
generators involving  $A_0^A$ and/or $\d/\d A_0^A$:
\begin{itemize}

\item realization on the space of polynomials in $A^A_i$ with coefficients in functions on $\fR^{d+1}$ with the origin excluded. In this case
\begin{equation}
A_0^A = X^A, \qquad \dl{A_0^A}=\dl{X^A}\,,
\end{equation}
where $X^A$ are Cartesian coordinates in  $\fR^{d+1}$. We keep the previous notation \eqref{glnot}, \eqref{smallTr}
for generators that do not involve $X^A$ and/or $\d/\d X^A$ while those that do are
denoted by
\begin{equation}
\label{sp3}
\ba{l}
\dps
\cS^\dagger_i=A_i^A\dl{X^A}\,,\qquad\; {{\bar \cS}}^i=X^A\dl{A_i^A}\,,
\\
\\
\dps
\cS^i=\dl{A^A_i}\dl{X_A}\,,\qquad \Box_X=\dl{X^A}\dl{X_A}\,.
\ea
\end{equation}
It is convenient to split the  $o(d-1,2)$ generators $J^
{AB}$ in two pieces as
$J^{AB} = L^{AB}+M^{AB}$, where an orbital part $L^{AB}$ is given by
\be
\label{orbit}
L^{AB} = X^A\dl{X_B} - X^B\dl{X_A}\;.
\ee

\item realization on the space of polynomials in $A_i^A$ with coefficients
in formal power series in variables $Y^A$ such that
\begin{equation}
\label{r21}
A_0^A = Y^\prime{}^{A} = Y^A+V^A\,, \qquad  \dl{A_0^A}=\dl{Y^A}\,,
\end{equation}
where $V^A$
is some $o(d-1,2)$ vector normalized as
$V^AV_A=-1$. Respective $sp(2n)$ generators are realized by
inhomogeneous differential operators on the space of functions
in $A^A_i$ and $Y^A$.  We use for them the following notation
\begin{gather}
\label{sp1}
\ba{l}
\dps
\sd_i=A_i^A\dl{Y^A}\,,\qquad \bsd^i=(Y^A+V^A)\dl{A_i^A}\,,
\\
\\
\dps
S^i=\dl{A^A_i}\dl{Y_A}\,,\qquad \Box_Y=\dl{Y^A}\dl{Y_A}\,.
\ea
 \end{gather}
Note that this realization is the same as in \cite{\AGT} but with $Y^A$
replaced by $Y^A+V^A$. Shifting by $V^A$ is crucial because this realization is inequivalent
with the usual one (\textit{i.e.}, the one with $V^A=0$). This happens because the change of variables $Y^A\to Y^A+V^A$ is ill-defined in the space of formal power series.
In contrast to the usual realization where highest (lowest) weight conditions
of $sp(2n-2)$ determine finite-dimensional irreducible $o(d-1,2)$-modules,
the inhomogeneous counterpart of these conditions can determine both finite-dimensional irreducible or infinite-dimensional indecomposable $o(d-1,2)$-modules. In particular,
it allows one to describe finite-dimensional gauge modules and infinite-dimensional Weyl modules
associated with gauge fields in AdS at the equal footing. Note that
the case $n=1,2$ has been originally described in \cite{\BGadS}. Analogous representation has been also used in~\cite{Bekaert:2009fg} to describe conformal fields.

The orbital part $L^{AB}$ of the generators $J^{AB}$ takes the form
\begin{equation}
\label{orbitalY}
L^{AB}  = (Y^A+V^A)\dl{Y_B} - (Y^B+V^B)\dl{Y_A} \;.
\end{equation}
This realization of the dual orthogonal and symplectic algebras will be refereed to as twisted Howe dual realization.

\end{itemize}

\subsection{Fields on the hyperboloid in terms of the ambient space}
\label{sec:howedualfields}

We start with the description of the unitary irreducible $o(d-1,2)$-modules
originally developed by Fronsdal \cite{Fronsdal:1979vb} for totally symmetric
fields and then extended to mixed-symmetry fields by Metsaev
\cite{Metsaev:1995re}. However, we need the description in terms
of a slightly different
basis for the irreducibility conditions and in terms of fields defined
on $\fR^{d+1}$ rather than on the hyperboloid.
We show that irreducibility conditions
imposed on fields on $\fR^{d+1}$ within the Metsaev formulation can be
seen as the highest weight conditions for an upper-triangular subalgebra of
$sp(2n)$. This is natural as $o(d-1,2)$ and $sp(2n)$ are dual in this representation in the sense of Howe duality. To make this algebraic interpretation manifest
we reformulate the Metsaev description using the basis elements \eqref{glnot}, \eqref{sp3} and restoring the radial dependence of the fields on the hyperboloid.

For the moment, we restrict our consideration to unitary massless fields \footnote{The case of
non-unitary massless fields as well as partially-massless fields in AdS space is  discussed in Sec. \bref{sec:beyond}.} in AdS space in the explicitly $o(d-1,2)$-invariant way. It is useful to define them as tensor fields on the ambient space $\fR^{d+1}$ with the origin excluded and the radial dependence eliminated through the
appropriate $o(d-1,2)$-invariant constraint.  More technically, tensor fields are represented by functions on $\fR^{d+1}/\{0\}$
taking values in the space of polynomials in variables $A^A_i, \,\,\,\,i=1,..., n-1,\,A=0,..., d$ introduced in Sec. \bref{sec:Preliminaries}.
Such a filed can be viewed  as a function $\phi = \phi(X, A)$. The radial coordinate dependence is effectively eliminated through the homogeneity condition
\begin{equation}
\label{radial}
\Big(N_X-k\Big)\phi=0\;,\qquad N_X = X^A\dl{X^A}\;,
\end{equation}
where $k$ is a number whose explicit value will be fixed later. This allows
to uniquely represent any field defined on hyperboloid in terms of the ambient space field satisfying \eqref{radial}. More explicitly, taking a new coordinate system $(r, x^m)$ in $\fR^{d+1}$, such that $r = \sqrt{-X^2}$ is a radius and $x^m$ are
dilation-invariant coordinates $N_X x^m = 0$, one finds $\phi=\phi_0(x,A)\, r^k$.

In order to describe irreducible representation let us impose the following
irreducibility conditions (in the sector of $A^A_i$-variables):
\begin{equation}
\label{algcond}
 T^{ij}\phi=0\,,\qquad N_i{}^j\phi=0\;\;\;\;\; i<j\,, \qquad (N_i-s_i)\phi=0\,.
\end{equation}
In addition, we also impose the
transversality, divergencelessness, and the "mass-shell" conditions
(in the sector of $X^A$-variables):
\begin{equation}
\label{difcond}
\bar \cS^i{}\phi=0\,,\qquad  \cS^i\phi=0\,,\qquad \Box_X \phi=0\,.
\end{equation}
In contrast to purely algebraic conditions \eqref{algcond} the latter ones
explicitly involve space-time coordinates.

All together, constraints \eqref{radial}, \eqref{algcond}, and \eqref{difcond}
form the upper-triangular subalgebra of $sp(2n)$ algebra supplemented with
Cartan elements.
Because $sp(2n)$ and $o(d-1,2)$ commute these constraints single out
an $o(d-1,2)$-module.\footnote{Note that this module is not necessarily irreducible. In the space of polynomials in $X^A$ these conditions are known to determine a finite-dimensional irreducible $o(d-1,2)$-module. This is not the case for functions on $\fR^{d+1}/\{0\}$ though.}
By solving the homogeneity condition \eqref{radial} and transversality constraints
$\bar \cS^i{}\phi=0$ one finds the description in terms of $o(d-1,1)$-tensor fields
defined on the hyperboloid as was originally observed in the case of totally symmetric
fields ~\cite{Fronsdal:1979vb}.

For fields subjected to the irreducibility conditions \eqref{radial} and \eqref{difcond},
one derives the following wave equation \cite{Metsaev:1995re}
\begin{equation}
 (\Box_{AdS} + m^2)\phi = 0\;, \qquad \Box_{AdS} \equiv \half L_{AB}L^{AB}\;,
\end{equation}
where $L_{AB}$ is an orbital part of $o(d-1,2)$ generators \eqref{orbit}.
Evaluating  $\dps\half L_{AB}L^{AB} =\Box_X -N_X(d-1+N_X)$
on the hyperboloid \eqref{hyper} and substituting  the mass-shell condition \eqref{difcond} one
finds an explicit value of the mass-like term
\be
\label{massa}
m^2 = N_X(d-1+N_X)\;.
\ee
Comparing with the original formula $m^2 = E_0(E_0+1-d)$
derived in \cite{Metsaev:1995re} we see that an eigenvalue
of the $o(d-1,2)$ energy operator $E_0$
and an eigenvalue of $sp(2n)$ Cartan element $N_X$ defined by \eqref{radial} are linearly dependent.

\subsection{Gauge invariance}
\label{sec:reductionMetsaev}

The theory determined by conditions \eqref{algcond}
and \eqref{difcond} does not in general describe irreducible fields.
More precisely, depending on the value of $N_X$ and $N_i$ the space of solutions may contain
singular vectors. In this case we obtain a gauge theory.

It is useful to describe the gauge symmetry using the BRST formalism.
To this end, we introduce  Grassmann odd ghost variables  $b_\alpha\,,\, \alpha= 1,..., p\, \leq\, n-1$. The BRST description comes together with the ghost number grading  $\gh{b_\alpha}=-1$ and  $ \gh{X^A} =\gh{A^A_i}=0$.  The gauge invariance is encoded in the
BRST operator
\begin{equation}
\label{brst-red0}
\brst_p=\cS_\alpha^\dagger\dl{b_\alpha}\,,
\end{equation}
acting in the space of functions $\Psi = \Psi(X,A|\, b)$ in $X^A$ taking values in polynomials in $A_i^A$ and ghost variables $b_\alpha$.

Field $\phi=\phi(X,A)$ considered in the previous section should be identified as the physical field which is the ghost-number-zero component of $\Psi(X,A|\, b)$ while the ghost number $-1$ component
is identified with gauge parameters. The gauge transformation is defined as
\begin{equation}
\label{gauge-t}
\delta \phi=\brst_p\chi\,,\qquad  \gh{\chi}=-1\,,\quad \gh{\phi}=0\,,
\end{equation}
where $\chi=\chi^\alpha(X,A) b_\alpha$ is a gauge parameter.

In order to consistently impose conditions~\eqref{radial}, \eqref{algcond}, and \eqref{difcond} some of them
are to be extended by ghost contributions to make $\brst_p$ act in the subspace. More precisely, in the ghost extended space one
imposes unchanged constraints
\begin{equation}
\label{cond-unmod}
{ \bar \cS}^i\Psi=0\;,
\quad
\cS^i\Psi = 0\;,
\quad
\Box_X\Psi = 0\;,
\quad
T^{ij}\Psi=0\;,
\end{equation}
and modified constraints
\begin{equation}
\label{cond-mod}
\hat\cN_i{}^j\Psi= 0\;\;\;i<j\;,
\qquad
(\hat\cN_i-s_i)\Psi= 0\,,
\qquad
\hat\cJ_\alpha{}^\beta \Psi = 0\;,
\end{equation}
where
\be
\label{DefConstraints}
\hat\cN_i{}^j = N_i{}^j+B_i{}^j\;,
\qquad
\hat\cN_i=N_i+B_i{}^i\,,\qquad
\hat \cJ_\alpha{}^\beta= \hat\cN_\alpha{}^\beta- \delta_\alpha^\beta\, (N_X
-B+ p+1)\;.
\ee
Here and in what follows we use the following useful notations:
\begin{equation}
 B_i{}^j =\delta_i^\alpha\, \delta^j_\beta\, b_\alpha\dl{b_\beta}\;,
\qquad  B_\alpha  = b_\alpha\dl{b_\alpha}\,,
\qquad
B= \sum_\alpha B_\alpha\,.
\end{equation}
Note that for ghost independent elements constraints
\eqref{cond-mod} impose additional restrictions compared to
their counterparts in \eqref{radial} and \eqref{algcond}. Additional constraints $\hat\cJ_\alpha{}^\beta$
appear as the consistency condition following from the commutators of the BRST operator $\brst_p$ with constraints ${\bar \cS}^\alpha$.

Requiring gauge invariance restricts possible values of
weights $s_i$. Their admissible values
are specified by consistency of the second and the third conditions in \eqref{cond-mod} which, in turn, originate from the consistency with the gauge transformation.
Indeed, from the third condition it follows that
$
\cN_\alpha{}^\beta \Psi = 0
$
for $\alpha \neq \beta$.
This implies $(\cN_\alpha- \cN_{\beta})\Psi=0$ and hence
$s_\alpha = s_\beta$ for all $\alpha, \beta$.
In other words, $\Psi$ has vanishing $sl(p)$ weights so that fields are $sl(p)$ singlets.
At the same time, they cannot
be $sl(p+k)$ singlets for $k>0$, therefore  $s_{p}> s_{p+k}$.
For the later convenience, we introduce a notation $s_1\equiv s$
and order the weights as follows
\be
\label{spinsarrangement}
s\equiv s_1=s_2 = ... = s_p>s_{p+1}\geq s_{p+2}\geq \cdots \geq s_{n-1}\;.
\ee

Consistency with the gauge transformation also fixes the value of constant term in the constraint $(N_X-k)\phi=0$ \eqref{radial} because it is now encoded in the constraint $\cJ_\alpha{}^\alpha$.
More precisely, for a physical field $\phi$  one gets
\be
N_X\phi = (s-p-1)\phi\;.
\ee
By virtue of formula \eqref{massa} one explicitly calculates a value of the
mass-like term
\be
\label{massaMetsaeva}
m^2 = (s-p-1)(s-p+d-2)\;,
\ee
thereby recovering the result of \cite{Metsaev:1995re} for unitary massless
AdS fields having the uppermost block of length $s$ and height $p$.

>From a more algebraic point of view, consistency with the gauge transformation \eqref{gauge-t}
extends conditions \eqref{radial}, \eqref{algcond}, and \eqref{difcond} that form
upper-triangular subalgebra of $sp(2n)$ (including Cartan elements)
to an extended set of conditions \eqref{cond-unmod} and \eqref{cond-mod} whose ghost independent parts form the parabolic subalgebra of $sp(2n)$. The consistency can be immediately seen from the fact that together with $\cS^\dagger_\alpha$ these conditions
also form a parabolic subalgebra.

To complete the description of unitary gauge fields in terms of the  ambient space
let us spell out in components the gauge transformation of the physical fields \eqref{gauge-t}.
It takes the following form
 \be
\label{pgauge}
\delta \phi = \cS^\dagger_1
\chi^1 + ... + \cS^\dagger_p \chi^p\;,
\ee
where gauge parameters $\chi^\alpha = \chi^\alpha(X,A)$ are components
of ghost-number $-1$ element  $\chi=\chi^\alpha b_\alpha$ satisfying constraints
\eqref{cond-unmod} and \eqref{cond-mod}.
The peculiar feature of the gauge transformation is that gauge parameters $\chi^\alpha$
do not satisfy Young symmetry conditions and are linearly dependent.
By virtue of constraints \eqref{cond-mod}
one can show that gauge parameters $\chi^\alpha$ at $\alpha<p$
are expressed through  parameter $\chi^p$ satisfying
Young symmetry conditions $N_i{}^j\chi^p = 0$, $i<j$ and weight conditions $N_\alpha \chi^p = (s-\delta_{\alpha p})\chi^p$. With the help of
gauge parameter $\chi^p$
gauge variation \eqref{pgauge} can be equivalently rewritten
in the form
\be
\delta \phi = \Pi \cS^\dagger_p \chi^p  \equiv (\cS^\dagger_p - \cS^\dagger_{p-1}N_p{}^{p-1}-\cS^\dagger_{p-2}N_p{}^{p-2}- \ldots - \cS^\dagger_{1}N_p{}^{1})\chi^p\;,
\ee
where $ \Pi $ involves appropriate Young symmetrizations needed to adjust symmetry properties of both sides \cite{Metsaev:1995re}.

\section{Generating BRST  formulation}

The formulation of the unitary gauge fields developed in the previous section
is not completely satisfactory. First of all, it is not a genuine local gauge field theory because gauge parameters are subjected to the differential constraints (\textit{i.e.}, constraints
involving derivatives with respect to $X^A$-coordinates). Furthermore, the way it is formulated
is not explicitly local because fields are defined in terms of the ambient space.
A natural question is to find a realization of the theory in terms of
internal coordinates on the hyperboloid and gauge parameters not subjected
to differential constraints. This can be done following the procedure used
in \cite{\BGadS} in the case of totally symmetric fields.

The idea suggested from~\cite{\BGadS}~\footnote{In its turn it originates in  (the generalization~\cite{\GL,\BGL,\BGST} to constrained systems of) the Fedosov quantization procedure~\cite{Fedosov:1994} and Vasiliev unfolded formalism~\cite{Vasiliev:1988xc,Vasiliev:1988sa,Vasiliev:2003ev}. In the related context it was also used in~\cite{Grigoriev:2006tt,Bekaert:2009fg}} is to put the ambient space
to the fiber of the vector bundle over AdS space and then eliminate additional degrees of freedom through auxiliary constraints.  More technically,
one replaces coordinates $X^A$ with formal variables $Y^A$ and then consider fields
on the AdS space with values in the fiber that is in the space of ``functions'' in $Y^A$ and $A^A_i$
variables. In this procedure all the algebraic constraints stay the same while
those involving $X^A$ and $\dl{X^A}$ (in particular, those entering the BRST operator)
are replaced with the respective constraints for $Y^A$ variables and hence also become  algebraic. The extra degrees of freedom are then eliminated by introducing
additional constraints.

\subsection{BRST operator and field equations}
\label{sec:interform}

The well-known approach to describe AdS geometry structure on manifold $\manX$
is to consider vector bundle $\cV_0$ over $\manX$ with the fiber being $(d-1,2)$- dimensional pseudo-Euclidean space. The AdS geometry structure is then encoded in the compatible flat $o(d-1,2)$-connection $\omega^{AB}(x)$ and a given section $V^A(x)$ of $\cV_0$ satisfying $\eta_{AB}V^AV^B=-1$, where $\eta_{AB}$ are coefficients of
the fiber-wise pseudo-Eucledean bilinear form. If in addition $\nabla V^A$ seen as a map from the tangent bundle to $\cV_0$ is of maximal rank then indeed $e^A=\nabla V^A$ can be identified with the vielbein and $\eta(e,e)$ with the AdS metric. Here $\nabla$ denotes the covariant derivative determined by connection $\omega$.

The space of polynomials in $A^A_i$ with coefficients in formal power series in $Y^A$ is equipped with the action of  $sp(2n)$ and $o(d-1,2)$ defined by \eqref{glnot}, \eqref{smallTr}, \eqref{sp1} and \eqref{orbitalY}, respectively. Taking this space as a fiber gives a vector bundle $\cV$ associated to $\cV_0$. In what follows, the fibre is also assumed
to contain ghost variables $b_\alpha$  on which $o(d-1,2)$ and $sp(2n)$ act trivially. The $o(d-1,2)$-connection $\omega^{AB}$ determines the following covariant derivative (also denoted by $\nabla$) in the associated bundle $\cV$
\begin{equation}
\label{nabla}
 \nabla=\derham+\half\theta^m\omega^{AB}_{m}J_{AB}\equiv
\theta^m\dl{x^m}-
\theta^m \omega^A_{m B}\Big((Y^B+V^B)\dl{Y^A}+A_i^B\dl{A_i^A}\Big) \,,
\end{equation}
where $\omega^{AB}_{m}$ and $V^A$ are components of $\omega^{AB}(x)$ and $V^A(x)$
introduced using a suitable local frame and $x^m$ are local coordinates on $\manX$.
Here the frame is chosen such that $V^A=const$; the expression for $\nabla$
gets additional terms if a local frame where $V^A\neq const$ is used. We have replaced
basis differential forms $dx^m$ with extra Grassmann odd ghost variables
$\theta^m,$ $\;m=0, ... , d-1$ because  $\nabla$ will be interpreted later as a part of
BRST operator.

Let us consider the following BRST operator
\be
\label{basicBRST}
\hat\brst=\nabla+Q_p\;,\qquad\quad Q_p=\sd_\alpha\dl{b_\alpha}\,,
\ee
defined on the space of sections of the bundle above.
%
We assign the following gradings to the ghost variables
$\gh{\theta^{m}}=-\gh{b_\alpha}=1$
so that BRST operator $\hat\brst$ has a standard ghost-number $\gh{\hat\brst}=1$.
The BRST operator is nilpotent  because of the following
obvious  relations~\footnote{Here and in what follows
the commutator denotes the graded commutator, $\commut{f}{g}=fg-(-)^{|f||g|}gf$,
where $|f|$ is the Grassmann parity of $f$.}
\be
\ba{c}
\label{blockcomrel}
\nabla^2 =Q_p^2 = 0\;,
\qquad
[Q_p, \nabla] =0\;.
\ea
\ee
The former relation holds in virtue of the zero-curvature
condition for connection $\omega^{AB}$. The  latter one is true because
$\nabla$ and $Q_p$ are build of generators of two commuting (Howe dual)
algebras $o(d-1,2)$ and $sp(2n)$.

That BRST operator is build out of the flat $o(d-1,2)$ connection and
$sp(2n)$ generators implies the explicit $o(d-1,2)$-invariance of the theory described by
$\hat\brst$.  To see how $o(d-1,2)$-algebra acts on fields let us note
that $o(d-1,2)$ naturally acts on the fibre at any point $x_0\in\manX$. This determines
an action on fields by taking as parameter a covariantly constant section of
the associated bundle with the fibre being $o(d-1,2)$ considered as the adjoint module.
\footnote{See \cite{\AGT} for a discussion of a general symmetry algebra and the example of Poincar\'e algebra}
 In terms of components, let $\xi^0_{AB}=-\xi^0_{BA}$ represent an $o(d-1,2)$-element.
It can be extended to a covariantly constant $\xi_{AB}(x)$ satisfying
$\nabla \xi_{AB}(x)=0$ and $\xi_{AB}(x_0)=\xi_{AB}^0$, where $x_0 \in \manX$ is a given point
of $\manX$. If $\phi=\phi(x,Y,A,\text{ghosts})$ represents a field then the $o(d-1,2)$-action
can be defined as
\begin{equation}
 R(\xi^0)\phi=
\half\xi_{AB}J^{AB}\phi
\,, \qquad \nabla \xi_{AB}(x)=0\,,\qquad  \xi_{AB}(x_0)=\xi^0_{AB}\,.
\end{equation}
Note that the above expression is not unique since it is defined modulo a
gauge transformation and terms proportional to the equations of motion.
For instance, one can also represent the action such that coordinates $x^m$ are affected
(see \cite{\AGT} for a more extensive discussion).

Let us recall how the BRST operator and its representation space encode
a gauge field theory. Physical fields are identified as
elements $\Psi^{(0)}$ at ghost number $ 0 $, gauge parameters as elements $\chi^{(-1)}$ at ghost number $-1$. The equations of motion and the gauge transformations read as
\be
\hat\brst \Psi^{(0)} = 0\;,\qquad \delta \Psi^{(0)} = \hat\brst \chi^{(-1)}\;,
\ee
where the gauge parameters have ghost number $\gh{\chi^{(-1)}}=-1$.
Elements at other ghost numbers correspond to higher structures of the gauge algebra.
For instance, order $k, ~~k = 1,..., p-1$ reducibility parameters are described by ghost-number $-k$ elements. The respective reducibility identities read as
$\delta \chi^{(-k)} = \hat\brst \chi^{(-k-1)}$.

Specializing to the case at hand: an element of vanishing ghost degree reads as
\begin{equation}
\label{diffforms}
 \Psi^{(0)}=\psi_0+\psi_1+\ldots+\psi_{p}\,,\qquad
\psi_k=\psi_{m_1\ldots\, m_{k}}^{\alpha_1\ldots\, \alpha_{k}}(x,Y,A) b_{\alpha_1}\ldots b_{\alpha_{k}}\theta^{m_{1}}\ldots \theta^{m_{k}}\,.
\end{equation}
The expansion coefficients $ \psi_{m_1\ldots\, m_{k}}^{\alpha_1\ldots\, \alpha_{k}} $ are identified as differential $k$-forms ($k\leq p$) on $\manX$. The equations of motion take the form
\begin{equation}
\label{interm-eq-motion}
\begin{aligned}
 \nabla \psi_0+\sd_\alpha\dl{b_\alpha}\psi_1&=0\,,&\\
 \nabla \psi_1+\sd_\alpha\dl{b_\alpha}\psi_2&=0\,,&\\
&\ldots &\\
\nabla \psi_{p}&=0\,.&
\end{aligned}
\end{equation}
First order gauge parameters can be represented as
\begin{equation}
 \xi^{(-1)}=\xi_1+\xi_2+\ldots+\xi_{p}\,,\quad \xi_k=\xi^{\,\alpha_1\ldots\, \alpha_k}_{\,i_1\ldots\, i_{k-1}}(x,A,Y)b_{\alpha_1}\ldots b_{\alpha_k}\theta^{i_1}\ldots \theta^{i_{k-1}}\;.
\end{equation}
For instance, gauge parameter  $\xi_1=\xi^\alpha b_\alpha$ is a $0$-form. The gauge transformations have the form
\begin{equation}
\label{gs}
\begin{aligned}
 \delta_\xi \psi_0&= \sd_\alpha\dl{b_\alpha}\xi_1\,,&\\
 \delta_\xi \psi_1&= \nabla \xi_1+\sd_\alpha\dl{b_\alpha}\xi_2\,,&\\
&\ldots &\\
 \delta_\xi  \psi_{p}&=\nabla \xi_{p}\,.&
\end{aligned}
\end{equation}
In the same way one can spell out the reducibility relations.

Let us stress that in this formulation the structure of the equations of motion and gauge symmetries is exactly the same as of the formulation~\cite{\AGT} for the Minkowski space fields. The difference
is in $o(d-1,2)$-module structure of the fiber replaced with the $iso(d-1,1)$
(i.e. Poincar\'e) one. Respectively, $o(d-1,2)$ covariant derivative
of the present formulation is replaced with Poincar\'e one. However, the explicit structure of the fibre is quite different for AdS and Poincar\'e gauge fields.
In particular, the algebraic constraints imposed to describe
irreducible fields belong to different algebras.

\subsection{Algebraic constraints}
\label{sec:algconstr}

The system just constructed does not describe an irreducible representation.
Moreover, it is an off-shell system in a sense that it does not impose true differential equations on fields. All equations are equivalent to constraints and can be solved in terms of some unconstrained fields (see~\cite{Sagnotti:2005ns,Vasiliev:2005zu,Grigoriev:2006tt}
for more details on the off-shell form of HS dynamics).
To make it dynamical one should impose the fiber version of the constraints \eqref{cond-unmod} and \eqref{cond-mod}. These read as
\be
\ba{c}
\label{ircon}
\dps
T^{IJ}\Psi=0\;,
\quad
\bsd{}^i\Psi=0\;,\quad
\dps\hat\cN_i{}^j\Psi\equiv(N_i{}^j+B_i{}^j)\Psi=0\;\;\;\;i<j\;,
\\
\\
\dps
\hat\cJ_\alpha{}^\beta\Psi\equiv (\hat\cN_\alpha{}^\beta
-\delta_\alpha^\beta(N_{Y^\prime}
-B+p+1))\Psi=0\;,
\ea
\ee
where we introduced Euler operator $N_{Y^\prime}=\dps(Y^A+V^A)\dl{Y^A}$ (cf.~\eqref{gl(p)}).
In addition, conditions
\be
\label{wcond}
\hat\cN_i\Psi\equiv (N_i+B_i)\Psi=s_i\Psi\,.
\ee
single out a particular spin field.
As before, all the constraints together with $S_\alpha^\dagger$ imposed through the BRST operator form a parabolic subalgebra of $sp(2n)$ represented on the fiber. This ensures the consistency  of the system. Note that among the constraints \eqref{ircon}, \eqref{wcond} those involving $\dl{Y^A}$ (except for $\hat\cJ_\alpha{}^\beta$)
lead to differential equations of motion while the remaining ones give rise to algebraic constraints.

Applying the same reasoning as in Sec. \bref{sec:reductionMetsaev}
one concludes that spins are arranged according to \eqref{spinsarrangement}.
In particular, it follows that constraints $\hat\cJ_\alpha{}^\beta$ split in two parts
\be
\label{irconsplit}
\hat\cN_\alpha{}^\beta\Psi = 0\;,
\qquad
h\Psi\equiv (N_{Y^\prime}
-B+p+1)\Psi=s\Psi\;,
\ee
for $\alpha\neq \beta$ and $\alpha=\beta$, respectively.

\subsection{Equivalence to Metsaev formulation}
\label{sec:equivalence}
Our next aim is to show that the theory determined by BRST operator \eqref{basicBRST} and the constraints \eqref{ircon} and \eqref{wcond} indeed describes unitary gauge fields. To this end
let us note that by eliminating auxiliary fields and fixing the gauge, equations of motion
\eqref{interm-eq-motion}
can be written as
\begin{equation}
\label{covconst}
\nabla \psi_0=0\,,\qquad  \nabla\psi_1+ \ldots=0\,, \qquad \nabla\psi_2+\ldots=0\,,\qquad\ldots\qquad\;,
\end{equation}
where by slight abuse of notation we denote by $\psi_k$ the field of the reduced theory.
More precisely, $\psi_k$ is a $k$-form obtained by eliminating auxiliary components and fixing algebraic gauge
symmetries from the respective fields in \eqref{diffforms} while dots in the equations for $\psi_k$ denote the extra terms depending
on $\psi_l$ with $l<k$.
%
This form of the equations of motion
is known as unfolded form~\cite{Shaynkman:2000ts,Vasiliev:2001wa}. It can be obtained~\cite{\BGST,\BGadS} from BRST formulation \eqref{interm-eq-motion} by reducing to  $Q_p$-cohomology. Here we do need an explicit form of the unfolded equations.
We only note that analysing the gauge invariance of \eqref{covconst} one concludes
that $\psi_0$ is invariant while higher components are determined in terms of $\psi_0$ modulo gauge transformations. It follows that physical degrees of freedom are carried by $\psi_0$ only.~\footnote{This is a general feature of the unfolded form of equations of motion~\cite{Bekaert:2005vh}.}

It is then enough to concentrate on equations for $\psi_0$ that decouple from others.
Field $\psi_0$ can be shown to take values in $Q_p$-cohomology
at vanishing ghost degree. Because the cocycle condition is trivial for ghost-number-zero elements the cohomology class can be identified with the equivalence class of $\psi_0$ from
\eqref{interm-eq-motion} modulo the equivalence relation
$\psi_0 \,\sim \,\psi_0+\sd_\alpha \chi^\alpha$.  Because $\nabla$ is flat there exists
local frame where connection coefficients $\omega$ vanish. In such frame equations for $\psi_0$ takes the form
\begin{equation}
\label{psi0exp}
\nabla\psi_0\equiv  \theta^m(\dl{x^m}-\ddl{V^A(x)}{x^m}\dl{Y^A})\psi_0=0\,,
\end{equation}
where we have reintroduced the term proportional to $dV^A$ that was missing in \eqref{nabla}
(because $V^A$ was assumed constant there). Moreover, in this frame the compensator components $V^A$ satisfying $V^2=-1$ can be identified with the Cartesian coordinates on the ambient space $\fR^{d+1}$ expressed through the intrinsic coordinates on $\manX \subset \fR^{d+1}$.
Note also that in this frame the interpretation of $\nabla$ as $o(d-1,2)$-connection is not straightforward.

On the other hand, let $\phi=\phi(X,A)$ be a field on the ambient space $\fR^{d+1}$ satisfying \eqref{radial}, \eqref{algcond}, \eqref{difcond} and subjected to gauge equivalence \eqref{pgauge} with the gauge parameters $\chi^\alpha$ satisfying \eqref{cond-unmod} and \eqref{cond-mod}. Let us introduce formal variables $Y^A$ and represent $\phi$  and $\chi^\alpha$ by $\psi(X,Y,A)$ and $\lambda^\alpha(X,Y,A)$ satisfying
\begin{equation}
\label{mapXY}
 (\dl{X^A}-\dl{Y^A})\psi=0\,, \quad \psi|_{Y=0}=\phi\,,\quad
(\dl{X^A}-\dl{Y^A})\lambda^\alpha=0\,, \quad \lambda^\alpha|_{Y=0}=\chi^\alpha\,.
\end{equation}
This representation is obviously one-to-one. \footnote{This is true both in the space of smooth functions and formal power series in $Y^A$-variables. As before we assume formal series.}
In view of \eqref{mapXY} one observes that $T \psi$ is equivalent to $\mathfrak T \phi$,
where $T$ and $\mathfrak T$ are two realization of an element of $sp(2n)$. More precisely,
$T$ and  $\mathfrak T$ are related by the change $X^A \leftrightarrow X^A+V^A$ and $\dl{X} \leftrightarrow \dl{Y^A}$
(see Section~\bref{sec:howe}). For instance, $\Box_X\phi$ is equivalent to $\Box_Y \psi$ if \eqref{mapXY} is imposed.

The condition $ (\dl{X^A}-\dl{Y^A})\psi=0$ can be interpreted as a covariant constancy condition $\nabla_0\psi=0$ with respect to an appropriate\footnote{This can be seen as a
standard $iso(d-1,2)$-connection.} connection $\nabla_0$ so that it is similar to \eqref{psi0exp}. Indeed, $\psi$ and $\psi_0$ take values in the same space of polynomials
in $A_i$ with coefficients in formal series in $Y^A$-variables. Although $\psi$ and $\psi_0$ are defined on different spaces ($ \fR^{d+1}$ and $\manX  $, respectively)  it turns out that $\nabla_0\psi=0$ and $\nabla \psi_0=0$
have isomorphic spaces of (equivalence classes modulo gauge invariance) solutions.
To see this let us first introduce a trivial vector bundle $\cV(\fR^{d+1})$ with the fiber
being the space of polynomials in $A_i$ with coefficients in formal series in $Y^A$-variables. It is associated to the tangent bundle over the ambient space $\fR^{d+1}$.

In terms of general coordinates $X^{\underline{A}}$ on $\fR^{d+1}$
the covariant derivative $ \nabla_0 $ takes the form
\cite{Barnich:2006pc} :
\be
\label{MainNabla}
\nabla_0 = \Theta^{\underline{A}}\Big(
\dl{X^{\underline{A}}}-\ddl{X^A}{X^{\ul{A}}}
\dl{Y^A}\Big)\;,
\ee
where new ghost variables $\Theta^{\ul A}$ stand for the basis differentials $dX^{\ul{A}}$. Note that using a general orthogonal local frame of the tangent bundle over $\fR^{d+1}$
would also bring the usual term $\Theta^{\underline{A}} W_{\underline{A}}{}^{AB}(Y_A\dl{Y^B}+A_{iA}\dl{A_i^B})$ with the connection coefficients
 $W_{\underline{A}}{}^{AB}$.
Furthermore, vector bundle $\cV(\manX)$ introduced in Section~\ref{sec:interform}
can be identified as a pullback of the bundle $\cV(\fR^{d+1})$ to $\manX\subset \fR^{d+1}$. Moreover, flat connection $\nabla$ in $\cV(\manX)$ can be seen as
a pullback of $\nabla_0$ in $\cV(\fR^{d+1})$ to $\cV(\manX)$. More explicitly, reducing to the hyperboloid amounts to choosing a new coordinate system
$(r, x^m)$ in $\fR^{d+1}$, where $r = \sqrt{-X^2}$ is a radial coordinate and
$x^m$ are dilation invariant coordinates. Then $\nabla$ can be seen as restriction
of $\nabla_0$ to the surface of fixed radius $r=1$ and radial ghost component $\theta_{(r)}=0$, and is given
by \eqref{nabla} if one identifies dilation-invariant coordinates and the intrinsic coordinates on $\manX$.

Restriction to $\manX$ clearly sends covariantly constant sections of $\cV(\fR^{d+1})$
to those of $\cV(\manX)$. Moreover, this map is an isomorphism. To see this, let us note that this would be a trivial statement if the fiber were finite-dimensional.
Indeed, a covariantly constant (with respect to $\nabla$) section defined at $r=1$ can be extended to a unique covariantly constant (with respect to $\nabla_0$) section defined in the vicinity of $r=1$. In the case at hand, however, the fiber is infinite-dimensional and
solving for $r$-dependence could result in a nonconvergent series. This does note happen because the $r$-dependence of $\phi$ is fixed by constraint $X^A\dl{X^A}\phi=k\phi$ \eqref{radial} which in turn originates from fiber constraint $(Y^A+V^A)\dl{Y^A}\psi=k\psi$.
This shows that restriction to $\manX$ is an isomorphism. In its turn, it determines an isomorphism between fields $\phi(X,A)$ satisfying \eqref{radial}, \eqref{algcond}, and \eqref{difcond} and the covariantly constant sections
$\psi_0(x,A)$ of $\cV(\manX)$ satisfying \eqref{psi0exp}.

This isomorphism is compatible with the $sp(2n)$ actions defined in Section~\ref{sec:howe}.
In particular, this guarantees that this map is compatible with the gauge transformation
so that spaces of respective equivalence classes are also isomorphic. In addition,
it is also compatible with the $o(d-1,2)$ action. This implies that
the value of the energy evaluated in Section \bref{sec:howedualfields}
remains the same.
Moreover, the computation of energy in Section \bref{sec:howedualfields} is only based on
the relations of $o(d-1,2)$ and $sp(2n)$ realized on the space of functions in $X^A,A^A_i$.
Because the relations are the same for realizations of the same algebras on the fiber
one immediately finds the same value for a fiber at a given point of $\manX$. As the equations of motion have the form of a covariant constancy conditions one finds that this value is the same everywhere for a given field configuration.

\subsection{Beyond the unitary case}
\label{sec:beyond}

We have by now constructed a compact gauge-invariant  description
of unitary gauge fields on AdS. It turns out that it can be generalized to a more general
class of fields. To demonstrate the idea of such a generalization let us first show how
the unitary fields can be seen as a subsector of a  wider theory.

Let us consider BRST operator that can be obtained from \eqref{basicBRST}
by taking $p=n-1$
\be
\label{fiberQ}
\brst = \nabla+Q\;,\qquad Q = \sd_i\dl{b_i}\;,\qquad i=1,\ldots, n-1\;,
\ee
acting on the subspace of a space of functions $\Psi=\Psi(x, Y,A|\,\theta, b)$ singled out by constraints
\be
\label{mainconstr}
T^{IJ}\Psi = 0\;,
\qquad
(N_i{}^j+B_i{}^j) \Psi = 0\;\;\;i<j\;,
\qquad
(N_i+B_i)\Psi = s_i\Psi\;.
\ee
Note that these form a subset of constraints \eqref{ircon} and \eqref{wcond}. As we are going to see this theory describes a reducible system so that one or another
irreducible field (not necessarily unitary) can be singled out by imposing further constraints.

For instance, suppose that in addition to \eqref{mainconstr} one impose the following constraints
\begin{equation}
\dl{b_i} \Psi = 0\;,
\qquad
i=p+1, \ldots, n-1\;.
\end{equation}
One then observes that in this subspace $\brst$ coincides with
\eqref{basicBRST} while constraints~\eqref{mainconstr} coincide
with their counterparts from~\eqref{ircon} and \eqref{wcond}.
Imposing the remaining constraints from~\eqref{ircon} and
\eqref{wcond} one indeed recovers the description of unitary
gauge fields presented in Sections~\bref{sec:interform} and
\bref{sec:algconstr}. This shows that unitary fields can indeed be
singled out from a big theory~\eqref{fiberQ} and
\eqref{mainconstr} through farther algebraic conditions.

To give an example of non-unitary fields let us take $n=2$ (totally symmetric fields) so that
$Q= \dps S^\dagger \dl{b}$ and impose the following constraints
\be
\label{pmfConstr}
(\bar S^\dagger)^t \Psi = 0\;,
\qquad
(N_{Y^\prime} -B + t+1)\Psi = s\Psi\,,
\ee
in addition to \eqref{mainconstr}. One can check that $Q$ indeed acts in the subspace.

To see which theory this defines let us evaluate $Q$-cohomology.
In the minimal ghost number the coboundary condition is trivial
so that the cohomology is defined by the cocycle condition
$S^\dagger \Psi_1 = 0$. This, in particular, implies that $\Psi_1$ is a polynomial in $Y^A$
and it is legitimate to re-express it in terms of change $Y^\prime{}^A =Y^A+V^A$.
The full list of conditions determining the cohomology at ghost degree $-1$ reads as
\be
\ba{c}
N\Psi_1 = (s-1)\Psi_1\;,
\quad
N_{Y^\prime} \Psi_1 = (s-t)\Psi_1\;,
\\
\\
(\bar S)^t \Psi_1 = 0\;, \qquad S^\dagger \Psi_1 = 0\;.
\ea
\ee
Being written in terms of variables $A^A$ and $Y^\prime{}^A$ these
give the description of cohomology classes in terms of two-row
$o(d-1,2)$ Young diagrams with the first row of length $s-1$ and
the second row of length $s-t$. These cohomology classes
determine gauge fields that are $1$-form connections with values
in the respective $o(d-1,2)$ module originally considered in
two-row Young diagrams \cite{Vasiliev:2001wa, Skvortsov:2006at}.
These are known to describe partially-massless dynamics of spin $s$ and depth $t$ field~\cite{Deser:2001pe,Deser:2001us,Zinoviev:2001dt,Skvortsov:2006at}.

The cohomology classes at vanishing ghost degree
can be represented by elements satisfying
\be
\mathsf{\bar S} \Psi_0 = 0\;,
\qquad
\mathsf{\bar S} = Y^A\dl{A^A}\;.
\ee
Comparing with constraints~\eqref{pmfConstr} gives the following generalized $V^A$-transversality condition
\be
\label{pmW}
V^{A_1}\cdots V^{A_m}\dl{A^{A_1}}\cdots \dl{A^{A_t}} \Psi_0 =0\;,
\ee
and
\be
(Y^A+V^A)\dl{Y^A}\Psi_0= (s-t-1)\Psi_0\;.
\ee
Along with $\mathsf{\bar S} \Psi_0=0$ this gives the description of the respective Weyl module.
Using the representation \cite{Barnich:2006pc} of the Weyl module
for massless spin-$s$ fields as a subspace of totally traceless elements satisfying
\be
Y^A\dl{A^A}\phi=0\;,\qquad
V^{A}\dl{A^{A}}\phi =0\;,
\qquad
(Y^A+V^A)\dl{Y^A}\phi= (s-2)\phi\;,
\ee
one finds that the partially-massless Weyl module \eqref{pmW} decomposes into a collection of  Weyl modules of massless Fronsdal fields of spins $s-t+1, \ldots, s-1, s$.

In this way we have extended the construction of previous sections to
partially-massless fields originally described in~\cite{Deser:2001pe,Deser:2001us,Zinoviev:2001dt,Skvortsov:2006at}.
In the similar manner, one can describe other irreducible AdS fields.
Indeed, by Howe duality basis elements of  $o(d-1,2)$ commutes with $Q$ and therefore AdS algebra acts in the $Q$-cohomology. The $Q$-cohomology
in non-zero negative ghost numbers $0<p\leq n-1$ has been explicitly calculated in \cite{\AGT}
and is represented by finite-dimensional irreps of $o(d-1,2)$ algebra. These give rise to $p$-form fields with values in the respective $o(d-1,2)$ irreps and coincide with those identified in \cite{Skvortsov:2009zu}. According to \cite{Skvortsov:2009zu} these fields correspond to all possible massless (unitary and non-unitary) fields and partially-massless fields of any symmetry types.

Let us finally comment on the relation of the theory determined by
\eqref{fiberQ} and \eqref{mainconstr} to  massless fields on
$\fR^{d+1}$. It turns out that this theory can be identified as a
pull-back to $\manX$ of a theory defined on $\fR^{d+1}$. This can
be constructed by considering fields on $\fR^{d+1}$ and replacing
$\nabla$ with $\nabla_0$ given by \eqref{MainNabla}. The resulting
theory describes massless mixed-symmetry fields~\footnote{Strictly
speaking one also needs to take $\fR^{d,1}$ rather then
$\fR^{d-1,2}$ in order to have a usual interpretation in terms of
representations of Poincar\`e group} propagating in $\fR^{d+1}$
spacetime \cite{\AGT,\BGST}. Indeed, the BRST operator and the
constraints simply coincide with those from~\cite{\AGT}. Under the
reduction to $\manX$ the massless fields on $\fR^{d+1}$ decomposes
into a collection of gauge fields propagating on $\manX$. As we
have seen on examples one or another irreducible subsystem can be
then singled out by auxiliary constraints compatible with
\eqref{fiberQ} and \eqref{mainconstr}. Let us note that this
ideology is to some extent analogous to that of
~\cite{Boulanger:2008up,Boulanger:2008kw} where the unfolded form
of the equations of motion for mixed-symmetry massless fields on
AdS has been constructed starting from massless fields on the
ambient space.

\section{Parent form and other formulations}
\subsection{Parent form}
Although the formulation constructed in Section~\ref{sec:interform} is rather compact
and transparent other formulations can also be useful. An efficient way to handle various forms of the theory is to start with a sufficiently wide formulation such that other ones can be seen as one or another particular reductions. Such a formulation is refereed to
as a
parent form of the theory and is known for
the case of totally symmetric~\cite{\BGST} and mixed-symmetry~\cite{\AGT} fields on Minkowski space as well as for totally symmetric AdS fields~\cite{\BGadS}.

A parent formulation for mixed-symmetry AdS fields can be constructed as follows:
introduce Grassmann odd ghost variables $c_i$ and $c_0$ associated to the constraints
$\Box_Y$ and $S^i$.  The total BRST operator reads then as
\be
\label{parBRST}
\brst^{parent}  = \nabla + \bar\brst\;,
\ee
where $\nabla$ is given by \eqref{nabla} and $\bar\brst$  is given by
\begin{equation}
\label{brstY}
\bar\brst=Q_p+\text{``more''}=\sd_\alpha\dl{b_\alpha}+c_iS^i+c_0\Box_Y-c_\alpha\dl{b_\alpha}\dl{c_0}\,.
\end{equation}
The representation space is that of formulation of Section~\bref{sec:interform}
extended by polynomials in new ghost variables $c_i,c_0$ and satisfying the following
constraints
\be
\label{const-parent}
\ba{c}
\mathbb{\bar S}^i\Psi=0\;,
\quad
\cT^{ij}\Psi=0\;,
\quad
\cN_i{}^j\Psi= 0\;\;\;i<j\;,
\quad
\cJ_\alpha{}^\beta \Psi = 0\;,
\quad
\cN_i\Psi=s_i\Psi\,,
\ea
\ee
where
\be
\label{1set}
\dps
\mathbb{\bar S}^i = {\bar S}^i +2C_0{}^i\;,
\quad
\dps
\cT^{ij} = T^{ij}+ G^{ij}\;,
\quad\dps
\cN_i{}^j = \hat \cN_i{}^j+C_i{}^j\;,
\ee
and
\be
\label{gl(p)}
\cJ_\alpha{}^\beta=
\hat \cJ_\alpha{}^\beta+\delta_\alpha^\beta(2 C_0 - C)\;.
\ee
Here in addition to $B$ and $B_i^j$ introduced above we have also used the
following useful notation for operators involving new ghost variables
\be
C_I{}^J = c_I\dl{c_J}\;,
\qquad
G^{ij} = \delta^i_\alpha\dl{c_j}\dl{b_\alpha}+\delta^j_\alpha\dl{c_i}\dl{b_\alpha}\;,
\ee
along with the respective Euler operators
\be
\dps
C_I = c_I\dl{c_I}\quad \text{(no summation)}\,,
\qquad
C = \sum_i C_i\,.
\ee

To see that this formulation is equivalent to the one of
Section~\bref{sec:interform} one introduces additional degree such
that $\deg{c_o}=\deg{c_i}=-1$ and reduces the theory to the
cohomology of the term $\brst^{parent}_{-1}=c_0\Box_Y+c_i S^i$
from \eqref{brstY} which carries lowest degree. In its turn, $\brst^{parent}_{-1}$ -cohomology is
concentrated in vanishing degree, and hence the reduced theory
coincides with the one of Section~\bref{sec:interform}.\footnote{See
\cite{\BGST,\BGadS} for more details on the equivalent reductions
in  cohomological terms}

\medskip

The constraints \eqref{const-parent} still contain those involving
$Y^A$ and $\dl{Y^A}$. These are ghost modified ${\bar S}^i$ and $h$ (recall that
$\cJ_\alpha{}^\beta$ can be split into $\cN_\alpha{}^\beta$ and $h$, cf. \eqref{irconsplit}).
It can be useful to implement these constraints through the BRST
operator with their own ghost variables so that only purely
algebraic constraints
\begin{equation}
\label{paralgconst}
\cT^{ij}\Psi=0,\quad \cN_i{}^j\Psi=0 \,\,\, i<j,\quad \cN_\alpha{}^\beta\Psi=0\,\,\,\alpha\neq\beta,\quad
(\cN_i-s_i)\Psi=0
\end{equation}
are directly imposed in the representation space. To show that
such formulation is equivalent to~\eqref{parBRST} one introduces a
degree such that the term involving ${\bar S}^i$ and $h$ is of
degree $-1$ and then reduces to its cohomology. This gives back
the theory~\eqref{parBRST}.~\footnote{This reduction is a
straightforward generalization of that from~\cite{\BGadS} to which
we refer for more details.}

\subsection{Ambient space parent theory}
\label{sec:ambpar}
The parent theory constructed in the previous section can be seen as a reduction
to the hyperboloid $\manX \subset \fR^{d+1}$ of the related theory defined on
the ambient space $\fR^{d+1}/\{0\}$. Indeed, the arguments analogous to those of Section~\bref{sec:equivalence} show that the theory determined by
\begin{equation}
 \brst^{parent~amb.}=\nabla_0+\bar\brst
\end{equation}
defined on $\fR^{d+1}$ can be reduced to that determined by \eqref{parBRST}. Here, $\nabla_0$
is the covariant derivative defined in~\eqref{MainNabla}. In addition,
as this theory is defined on the entire $\fR^{d+1}/\{0\}$ one also needs to replace
components of $V^A$ with Cartesian coordinates $X^A$ on $\fR^{d+1}/\{0\}$
in the expression of constraints.

Because all the constraints involving $Y^A$ and $\dl{Y^A}$ can be
assumed to be imposed through the BRST operator one can
consistently eliminate variables $Y^A$ and $\Theta^A$. Indeed, using
Cartesian coordinates on $\fR^{d+1}$ and an appropriate degree one
identifies $\Theta^A\dl{Y^A}$ as a lowest degree term in the total
BRST operator. Because $Y^A$ and $\Theta^A$ are unconstrained
variables the cohomology can be identified with
$Y^A,\Theta^A$-independent elements (see \cite{\BGST}, where the analogous reduction was discussed in more details).
Under this reduction all the remaining operators are changed
according to $Y^A+X^A \to X^A$ and $\dl{Y^A} \to \dl{X^A}$ so that
the reduced BRST operator reads as
\begin{equation}
\label{brst}
\brst^{ambient}=c_0\Box_X+c_i\cS^i+\cS_\alpha^\dagger\dl{b_\alpha}-c_\alpha\dl{b_\alpha}\dl{c_0}\;,
\end{equation}
Here we assumed that BRST invariant extensions of constraints
$\bar\cS^i$ and $h_X$ are imposed directly. All the algebraic
constraints \eqref{paralgconst} stay the same.

Let us analyze the resulting ambient space theory in some more details.
Fields $\Psi = \Psi(X,A|\,b,c)$ are convenient to represent
in the form of the decomposition $\Psi = \Psi_1+ c_0 \Psi_2$.
For the ghost-number-zero component $\Psi^{(0)}$ fields
$\Psi_1^{(0)}\equiv \Phi$ and $\Psi_2^{(0)}\equiv C$ are the following decompositions with respect to the ghost variables
\be
\label{sf}
\ba{l}
\dps
\Phi =
\sum_{k=0}^{p}\;c_{i_1} \cdots c_{i_k} b_{\alpha_1} \cdots b_{\alpha_k}\,
\Phi^{i_1...i_k|\alpha_1...\alpha_k}\;,
\\
\\
\dps
C=\sum_{k=0}^{p-1}\;c_{i_1} \cdots c_{i_k}
b_{\alpha_1} \cdots b_{\alpha_{k+1}}\,
C^{i_1...i_k|\alpha_1...\alpha_{k+1}}\;.

\\
\ea
\ee
The expansion coefficients in \eqref{sf} are antisymmetric in each group of indices
 and the slash $|$ implies that no symmetry properties between
two groups
are assumed. In other words, the expansion components take values in
tensor products of $gl(n-1)$ and $gl(p)$ antisymmetric
irreps.
Note that these component fields can be seen
as an AdS version of the generalized triplets  discussed in
\cite{Bengtsson:1986ys,Sagnotti:2003qa,Bonelli:2003kh,\AGT}.

Decomposing the BRST operator with respect to the homogeneity degree in $c_0$
as $\brst^{ambient} = \dps\brst_{1}+\brst_0+\brst_{-1}$ one can reduce
the original theory to
the cohomology of $ \brst_{-1} = \dps c_\alpha\dl{b_\alpha}\dl{c_0}$
(see \cite{\BGST,\AGT} for details). One then concludes that fields $C$
are auxiliary while some components of $\Phi$ are Stueckelberg.
After the reduction one is left with the fields annihilated by
operator $Z_+=\dps c_\alpha\dl{b_\alpha}$ which is naturally interpreted as
a generator of $sl(2)$ realized on ghosts (see \cite{\AGT} for an explicit
discussion of this issue in the similar context).
This reduction provides a relationship between
the AdS version of generalized triplet formulation and the
ambient space metric-like formulation. In particular, one can show that subjecting
the dynamical fields of the reduced theory to
the BRST extended trace conditions
yields the generalized double-tracelessness conditions introduced in \cite{Alkalaev:2003qv}.

\section{$Q_p$ -  cohomology and BMV conjecture}
\label{sec:weyl}

For the sake of completeness, we show here that the constructed generating
formulation reproduces infinite-dimensional Weyl module and finite-dimensional
module of gauge fields of
the unfolded formulation for AdS mixed-symmetry massless fields \cite{Alkalaev:2003qv,Boulanger:2008up,Boulanger:2008kw}. More precisely,
the $Q_p$ - cohomology in the zeroth ghost degree is identified as Weyl module,
while $Q_p$ - cohomology in the minimal ghost number $-p$ is identified
as the gauge module. In all other ghost degrees the cohomology is empty.
Representation of the Weyl module as $Q_p$-cohomology allows to describe it in terms
of Lorentz irreducible fields that become Minkowski space gauge fields in the flat limit. In particular, this gives a proof of the Brink-Metsaev-Vasiliev (BMV) conjecture,
put forward in \cite{Brink:2000ag}, and partially proved in \cite{Boulanger:2008up,Boulanger:2008kw}.

\subsection{Elimination of $(d+1)$-th direction}
So far we used manifestly $o(d-1,2)$ covariant language. In this section
it is convenient to analyse the problem in terms of Lorentz (\textit{i.e.} $o(d-1,1)$) tensor fields.
To this end, we choose the local frame where $V^A = \delta^A_{d}$.
Set $Y^a = y^a$ and  $Y^{d} =z$. Analogously, $A^a_i = a_i^a$ and
$A^d_i = w_i$. In what follows, we always assume that all elements $\Psi = \Psi(Y,A|\,b)$
are totally traceless, $T^{IJ}\Psi = 0$.
The following statement shows how constraints
${\bsd}{}^i$ and $h$ from \eqref{ircon}, \eqref{irconsplit} eliminate the dependence on $(d+1)$-th variables $z$ and $w_i$.

\begin{prop}
\label{prop:isom}
 The space of all totally traceless elements $\Psi = \Psi(Y,A|\,b)$ satisfying
 \begin{equation}
\label{bsdh}
{\bsd }{}^i \Psi=0\,,\qquad (h+m)\Psi=0\,,
 \end{equation}
is isomorphic to the space of all $z,w_i$-independent totally traceless
elements. Here $m$ denotes any integer. The isomorphism sends $\Psi$ to the
traceless component of $\Psi|_{z=w_i=0}$.
\end{prop}
The dependence of elements on ghost variables $b_\alpha$ is inessential here and is introduced
for future convenience.
\begin{proof} The proof is a straightforward generalization of that
from~\cite{\BGadS}. The idea is to introduce auxiliary differential
$$
\delta=\gamma_i {\bsd}{}^i+\alpha (h+m)-\alpha\gamma_i\dl{\gamma_i}\,,\qquad \quad \delta^2=0\,,
$$
where $\gamma_i,\alpha$ are auxiliary Grassmann odd ghost variables,
$\gh{\gamma_i} = \gh{\alpha}=1$. For a ghost-number-zero
element $\Psi$, equation $\delta \Psi=0$
is equivalent to equations \eqref{bsdh}.
More formally, such elements can be identified with $\delta$-cohomology
at vanishing ghost number.

The statement amounts to showing that any traceless $z,w_i$-independent
$\Psi(y,a|\,b)$ can be uniquely completed to a totally traceless element
annihilated by $\delta$. If one takes homogeneity in $z,w_i$ as a degree
such a completion can be constructed order by order using the homological perturbation theory. More precisely, decomposing $\delta$ according to the degree
\be
 \delta=\delta_{-1}+\delta_0\;,\qquad \delta_{-1}=\alpha\dl{z}+\gamma^i\dl{w^i}\;,
\ee
one observes that such a completion exists and is unique provided
$\delta_{-1}$-cohomology is trivial (any $z,w_i$-independent elements). This is obviously the case
in the space of all (not necessarily traceless) elements. That this is also the
case in the traceless subspace is a straightforward generalization of the
respective statement proved in~\cite{\BGadS}.
\end{proof}

Both the space of $z,w_i$-independent traceless elements and its isomorphic space
are $sl(n-1)$-modules (in fact, also $gl(n-1)$-modules), with $sl(n-1)$ algebra generated by operators
$N_i{}^j=\dps A^A_i\dl{A^A_j},\,\,\, i\neq j$ and $\smalln_i{}^j=\dps a^a_i\dl{a^a_j},\,\,\, i\neq j$, respectively. The isomorphism above is also an
isomorphism of $sl(n-1)$-modules. Indeed,
\begin{equation}
\cP((N_i{}^j\Psi)|_{z=w=0})=\smalln_i{}^j(\cP(\Psi|_{z=w=0}))\;,
\end{equation}
where $\cP$ denotes the standard projector to a totally traceless component.
That the spaces above are isomorphic as $sl(n-1)$-modules implies,
in particular, that if $m=s$ then the
subspace of \eqref{bsdh} satisfying in addition irreducibility
conditions~\eqref{ircon}  is isomorphic to a subspace of traceless
$z,w_i$-independent elements satisfying the respective constraints
in terms of $\smalln_i{}^j$. One may formulate the above statement as
follows: when reducing to Lorentz all constraints remain intact
while $\bar S^i$ and $h$ are relaxed. In particular, all the weights $s_i$
remain the same.

Furthermore, the action of the BRST operator can be represented in terms of
$z,w_i$-independent elements using the isomorphism of Proposition~\bref{prop:isom}.
It is easy to check that
 \begin{equation}
\cP((Q_p\Psi)|_{z=w=0})=\smallq_p \cP(\Psi|_{z=w=0})
\,,\quad \text{where} \quad
\smallq_p=\smalls^\dagger_\alpha\dl{b_\alpha}\equiv a^a_\alpha\dl{y^a}\dl{b_\alpha}\,.
 \end{equation}
This implies that the field theory determined by
$\brst=\nabla_0+Q_p$ can be completely reformulated in terms of
$z,w_i$-independent fields. In these terms the respective  BRST
operator reads as
\be
\label{Minbrst}
\tilde \brst= \tilde \nabla + \smallq_p\;,
\ee
where $\tilde \nabla$ represents the action of $\nabla$ in terms of
$z,w_i$-independent fields. It acts
in the space of totally traceless functions $\phi = \phi(x, y, a|\,b, \theta)$ subjected
to the following conditions
\be
\label{Minirrcond}
(\smalln_i+B_i)\phi=s_i\,\phi\,,
\quad
(\smalln_i{}^j+B_i{}^j)\phi=0\,\,\, i<j\,,\quad
(\smalln_\alpha{}^\beta+B_\alpha{}^\beta)\phi=0
\ee
where spins are arranged as in \eqref{spinsarrangement}.
Although this form of the theory is not very useful
because the explicit expression of $\tilde\nabla$ and hence the form of the equations of motion is rather involved in terms of $o(d-1,1)$-tensor fields we are going to use it for the analysis of the spectra of unfolded fields. These can be found as $Q_p$-cohomology classes.

In the flat limit $\tilde\nabla$ becomes $\dps \tilde
\nabla\Big|_{\Lambda=0} = \theta^a(\dl{x^a}-\dl{y^a})$, where we made use of
standard flat coordinates $x^a$ and the associated local frame. Remarkably, in this limit the theory describes a
dynamics of a particular collection of Minkowski mixed-symmetry
fields. Indeed, \eqref{Minbrst} coincides with the BRST operator from \cite{\AGT}
describing mixed-symmetry Minkowski fields provided one replaces $\tilde\nabla$
with a usual flat Poincar\'e covariant derivative.
Moreover, for rectangular fields ($p=n-1$) conditions \eqref{Minirrcond}
explicitly coincides with their counterpart from \cite{\AGT} so that
in this case  the flat limit is simply identical with the respective
Minkowski field. More generally, if $p<n-1$ the flat limit
of the theory \eqref{Minbrst} has less gauge invariance (only $s^\dagger_i$ with $i\leq p$
determine gauge symmetry) then its Minkowski space counterpart and hence carries more degrees of freedom.
The fact that in the flat limit an irreducible  AdS gauge field decomposes into
a collection of Minkowski fields is known as BMV conjecture~\cite{Brink:2000ag}.
In Section~\bref{sec:BMV} we give a general
proof of the conjecture for fields of any symmetry type. Note that for fields
with at most four rows a correctness of the BMV conjecture
has been recently established in \cite{Boulanger:2008kw, Boulanger:2008up}.

\subsection{AdS Weyl module}
\label{sec:Weyl}
If one reduces the theory to $Q_p$-cohomology, elements of vanishing ghost number give rise to gauge invariant fields that are zero forms. In the literature the module where
these fields take values is known as Weyl module. In the present context we have the following:
\label{sec:zerothgh}
\begin{definition}
\label{def:Weyl}
An AdS Weyl module $\tilde\modM_0$ of spins $s_1,\ldots, s_{n-1}$ is a ghost number zero $Q_p$-cohomology evaluated in the subspace of elements $\phi(Y,A|b)$ satisfying~\eqref{ircon} and \eqref{wcond}.
\end{definition}
At ghost number zero the cocycle condition is trivial, while the coboundary condition says that any element of the form $\sd_\alpha \chi^\alpha$ is trivial.
As was explained above the $Q_p$-cohomology can be computed as
cohomology of $\smallq_p = \smalls^\dagger_\alpha\dl{b_\alpha}$
in the subspace of $z,w_i$-independent traceless elements satisfying~\eqref{Minirrcond}.

Before the actual analysis of the $q_p$-cohomology let us first
introduce some useful notation and definitions. As only generators
of $sl(n)$ algebra are involved in the constraints and the BRST
operator, it is enough to compute cohomology in the subspace
$\cK^{(k)}$ of traceless homogeneity-$k$ polynomials in
$a^a_i,y^a$ tensored with ghost variables, \textit{i.e.} the
respective eigenspace of the Euler operator $\smalln=\dps
\smalln_y + \sum_i \smalln_i$, $\; n_y = y^a\dl{y^a} $. Indeed,
all $sl(n)$ generators do not change the homogeneity degree.

In its turn, $\cK^{(k)}$ decomposes into a collection of  finite-dimensional irreducible
$sl(n)$-modules. Obviously, the following sets
\be
\label{sl(n-1)}
n_- = \Big\{ \smalln_i{}^j \;\;\; i<j \Big\}
\quad \text{and}\quad
n_+ = \Big\{ \smalln_i{}^j\;\;\; i>j \Big\}\;,
\ee
generate $sl(n-1)\subset gl(n)$ subalgebra and can be identified as
the upper-triangular and the lower-triangular subalgebras of $sl(n-1)$.

In order to realize the AdS Weyl module in terms of representatives of the equivalence classes it is
useful to restrict the analysis to a finite-dimensional irreducible $sl(n)$-module $V\subset \cK^{(k)}$.
In particular, module $V$ is completely specified by the eigenvalues
$m_y(\psi_0),m_i(\psi_0)$ of its highest weight (HW) vector $\psi_0$
with respect to the Euler operators $n_y,n_i$.

Conditions $\smalln_{i}{}^{j}\phi=0$ for $i<j$ imposed on $\phi$
are in fact the HW conditions with respect  to $sl(n-1)$ subalgebra. The space of $n_-$-invariant elements can be then seen as a subspace $V_0 \subset V$ of $sl(n-1)$ HW vectors.
Decomposing $V$ into the irreducible $sl(n-1)$-submodules
as
\begin{equation}
 V=\bigoplus_i V_i\;,
\end{equation}
and using the natural projection to the $sl(n-1)$ HW subspace of any irreducible
$sl(n)$-module, one defines the projector $\Pi : V\to V$ such that $\Pi^2=\Pi$ and $\im \Pi$
is the $n_-$-invariant subspace.

We have the following two Lemmas. Integers $m_i(\phi)$ below are eigenvalues
of the Euler operators $\smalln_i$ acting on $\phi$.

\begin{lemma}
\label{lemma:1}
Let $\phi$ be an $sl(n-1)$ HW vector from $V_i\subset V$ then $\phi$ can be represented as
\begin{equation}
\label{repSL}
 \phi=\Pi\smalls^\dagger_{i_1}\ldots\smalls^\dagger_{i_l} \Lambda^{i_1\ldots i_l} \psi_0\,,
\end{equation}
where $\psi_0$ is a HW vector of irreducible $sl(n)$-module $V$ and $\Lambda^{i_1\ldots i_l}$ are some coefficients.
\end{lemma}

\begin{lemma}
\label{lemma:2}
Let $\phi$ be an $sl(n-1)$ HW vector from $V_i\subset V$ then the conditions
$m_\alpha(\phi)=m_{\alpha}(\psi_0)$ and $\phi\neq \smalls^\dagger_\alpha \chi^\alpha$are equivalent.
\end{lemma}
Both Lemmas follow from basic properties of finite-dimensional irreducible $sl(k)$-modules (see Appendix \bref{sec:A}).
Lemma \bref{lemma:2} gives a description of $Q_p$-cohomology at zeroth ghost numbers in terms of
HW vectors of irreducible $sl(n)$-modules.

\subsection{AdS Weyl module in terms of Poincar\'e ones: BMV conjecture}
\label{sec:BMV}
Let us recall that a Poincar\'e Weyl module of spin $l_1\geq l_2,\ldots \geq l_{n-1}$ \cite{Skvortsov:2008vs} can be defined as  a subspace of
$sl(n)$ HW vectors in $\cK$ satisfying the respective weight conditions
(see \cite{\AGT} for more details).  It turns out that the AdS Weyl $\tilde\modM_0$ can be decomposed into the direct sum
of some Poincar\'e Weyl modules.

More precisely, given AdS Weyl module of spin
\eqref{spinsarrangement} a Poincar\'e  Weyl module is called admissible
associated module if $l_i=s_i-\nu_i$, where nonnegative integers $\nu_i = 0$, $ i\leq p$ and $ \nu_i \neq 0 $,  $ i>p $, are
chosen in a way compatible with the Young symmetry.
We have:
\begin{prop}
 AdS Weyl module $\tilde\modM_0$ of a given spin is isomorphic to a direct sum
of the admissible associated Poincar\'e Weyl modules.
\end{prop}
\begin{proof} We prove the statement by constructing the isomorphism explicitly.
Let us first restrict to $\cK^{(k)}$. As usual, we decompose $\cK^{(k)}$ into the direct sum of irreducible $sl(n)$-modules. Let $V$ be a given irreducible component.
Its highest weight vector $\psi_0$ by definition belongs to some Poincar\'e Weyl module.
Two things can happen: either  $\psi_0$ is admissible or not. If not then in $V$ there are no elements from $\tilde\modM_0$. If $\psi_0$ is admissible then there are nonnegative integers $\nu_i$ such that $s_i=l_i+\nu_i$ and $\nu_i=0$ for $i\leq p$.
It then follows from Lemma~\bref{lemma:2} that
\begin{equation}
\phi=\cI_{V}(\psi_0)=\Pi \left[ (\smalls^\dagger_{p+1})^{\nu_{p+1}}\ldots (\smalls^\dagger_{n-1})^{\nu_{n-1}}
\psi_0 \right]
\end{equation}
belongs to $\modM_0$. Note that $\phi$
is the only element in $V$ that belongs to $\tilde\modM_0$. Defining the map
$\cI_{V}$ for each irreducible $V$ (if $V$ is not admissible $\cI_V$ is trivial)
one determines $\cI$ for any element of the Poincar\'e module. By construction,
$\cI_V$ is an isomorphism.
\end{proof}

\subsection{AdS Gauge module}

To complete the description of the spectrum of unfolded fields let us identify the
cohomology at negative ghost degrees. The respective fields
take values in the so-called gauge module. At ghost degree $-p$ the fields are identified as
differential $p  $-forms taking values in the respective $ o(d-1,2) $ modules ~\cite{Alkalaev:2003qv}.
Namely, the coboundary condition
is trivial while the cocycle implies that $\sd_\alpha\phi=0$, where
$ \phi  = \phi_{m_1 \ldots\, m_p}\theta^{m_1}\cdots \theta^{m_p} $ takes values in a
subspace singled out by constraints \eqref{ircon}, \eqref{wcond}. In particular,
the field $\phi$ fulfills the following  conditions:
$\bar S^i \phi = 0 $ for all $i$ and  $(N_\alpha-s+1)\phi=0$ and
$(N_i-s_i)\phi=0$ for $ i>p $.
In view of these conditions representatives can be chosen polynomials in $Y^A$.
In terms of $Y^A{}^\prime=Y^A+V^A$ all the conditions give an explicit characterization
of gauge modules in terms of $o(d-1,2)$ Young tableaux having the uppermost block of length $ s-1 $ and height $ p+1 $~\cite{Alkalaev:2003qv}.
It turns out that $Q_p$-cohomology at ghost numbers other than $0,-p$ vanish.
To see this we again use the representation in terms of $z,w_i$-independent elements.

First of all we note that  constraints $\hat \cN_\alpha{}^\beta\psi = 0$ \eqref{irconsplit}
for $\alpha\neq \beta $  imply that element
$\psi = \psi^{\alpha_1 \ldots \alpha_k}b_{\alpha_1}\ldots b_{\alpha_k}$ with fixed
weights contains just
one independent component $\phi_{(k)}\equiv \psi^{p-k+1\;\ldots \;p-1\;p}$ satisfying
$sl(n-1)$ HW conditions $N_i{}^j \phi_{(k)} = 0$, $ i<j$.
The Young tableau associated to $\phi_{(k)}$ includes the uppermost block
of size $ [s,p-k] $, the neighboring  block of size $ [s-1, k] $, while the rest of the diagram has rows of lengths $ s_i $.

Operator $ \smallq_p$ obviously acts in the space of $sl(n-1)$ HW elements of definite weights. More precisely,  $  \smallq_p:\, \phi_{(k+1)} \mapsto \phi_{(k)} = \Pi \smalls^\dagger_{p-k}\phi_{(k+1)}$, where $\Pi$ is a projector on $sl(n-1)$ HW elements (see Sec. \bref{sec:zerothgh}).

For the $ \smallq_p$-cohomology at ghost number $-k$ we have the following cocycle
and the coboundary conditions:
\begin{equation}
\label{neskoro}
\Pi \smalls^\dagger_{p-k+1}\phi_{(k)} = 0\;,\qquad
%
\phi_{(k)} \sim \phi_{(k)}+ \Pi \smalls^\dagger_{p-k}\chi_{(k+1)}\;,
\end{equation}
where $\chi_{(k+1)}$ are some $sl(n-1)$ HW elements of definite weights.
Note that for $ k=0 $ the cocycle condition is trivial as was already discussed in Section~\bref{sec:Weyl}.
For $ k=p $ the coboundary condition is missing so we are left with
the cocycle condition only.  For  intermediate values of
the ghost number $ 0<k<p $ we have the following Lemmas describing solutions to \eqref{neskoro}.

\begin{lemma}
\label{lemma:3}
Let $\phi_{(k)}$ be an $sl(n-1)$ HW vector from $V_i\subset V$ then
conditions $m_\alpha(\phi_{(k)})=m_{\alpha}(\psi_{0})$ at $ 1\leq \alpha \leq p-k$, and $\phi_{(k)} \neq \Pi \smalls^\dagger_{p-k}\chi_{(k+1)}$ are equivalent.
\end{lemma}

\begin{lemma}
\label{lemma:4}
Let $\phi_{(k)}$ be an $sl(n-1)$ HW vector from $V_i\subset V$ then
$\Pi \smalls^\dagger_{p-k+1}\phi_{(k)} = 0$ iff $m_\alpha(\phi_{(k)})>m_{\alpha}(\psi_{0})$
for some $\alpha$ such that $ 1\leq \alpha \leq p-k$.

\end{lemma}
Here $\psi_0$ denote respective $sl(n)$ HW vectors from Lemma \bref{lemma:1}.
Both Lemmas result from comparing admissible weights of $ sl(n-1) $ HW elements
$\phi_{(k)}$ and their associated $ sl(n) $ HW elements $ \psi_{0} $ (see Appendix \bref{sec:A}).
Since there are no $ sl(n-1) $ HW elements that simultaneously satisfy both the cocycle and the coboundary conditions,  one concludes that the cohomology is empty for $ k\neq 0,p$.

\section{Conclusions}
\label{sec:conclusion}

In this paper, we have proposed the unified  formulation for unitary dynamics
of free bosonic HS fields of any symmetry type in the AdS space.
We have also observed and discussed how to generalize the theory to include
non-unitary fields. In particular, we have explicitly described such a generalization
for totally symmetric partially-massless fields.
The theory is formulated on the level of equations of motion using the usual BRST first quantized language. This makes the formulation somewhat analogous to the usual string-inspired BRST approach to higher spin fields. In particular,
this can make the proposed formulation useful in describing relation to (a tensionless limit of) the bosonic string theory on the AdS background.

Another motivation and possible application of these results have to do with
studying consistent interactions for mixed-symmetry AdS fields.
While in the case of totally symmetric fields
consistent interactions are known to cubic order in the Lagrangian
formulation~\cite{Fradkin:1987ks, Vasiliev:2001wa, Alkalaev:2002rq,Metsaev:2006zy,Buchbinder:2006eq,Zinoviev:2008ck,Boulanger:2008tg,Manvelyan:2009tf} and to all orders
at the level of equations of motion~\cite{Vasiliev:1992av,Vasiliev:2003ev},
interactions of mixed-symmetry AdS gauge fields are not known so far.
We hope that the
transparent algebraic structure and a due control of the gauge invariance through the
BRST technique make the present formulation useful in searching for nonlinear theory.
Moreover, a possible nonlinear deformation is necessarily related to the appropriate
algebraic structure -- higher spin algebra. In the case of totally symmetric fields
the respective algebra~\cite{Vasiliev:2003ev,Vasiliev:2004cm} can be identified with
higher symmetries~\cite{Eastwood:2002su} of the scalar singleton, the
corresponding  algebra in the mixed-symmetry case is expected to be related to singletons of
nonvanishing spins.
The respective candidate higher spin algebras have been recently identified in~\cite{Bekaert:2009fg} using a framework closely related to
the present one (see also a discussion of singleton composites in ~\cite{Boulanger:2008kw}).

\vspace{5mm}

\textbf{Acknowledgments.}
We are grateful to I. Tipunin for many fruitful discussions. We also acknowledge discussions
with G. Barnich, N. Boulanger, E. Feigin, C. Iazeolla, R. Metsaev,
O. Shaynkman, E. Skvortsov,  P. Sundell and M. Vasiliev.
This work is supported by the LSS grant Nr 1615.2008.2.
The work of  KA is supported in part by grants
RFBR grant Nr 08-02-00963 and the Alexander von Humboldt Foundation grant PHYS0167.
The work of MG is supported by the RFBR grant 08-01-00737 and RFBR-CNRS grant 09-01-9310.

\appendix
\section{Proofs of Lemmas of Section \bref{sec:weyl}  }
\label{sec:A}

\paragraph{Proof of Lemma \bref{lemma:1}.}

Any element from $V\subset \cK^{(k)}$ can be represented as a linear combination of elements
obtained by acting on $\psi_0$ with $n_+$ and $\smalls^\dagger_i$.
Representing $\phi$ in this way, moving $n_+$ to the left by using the algebra commutation
relations,
and applying $\Pi$ one finds that $\phi=\Pi\phi=\Pi \smalls^\dagger_{i_1}\ldots\smalls^\dagger_{i_l} \Lambda^{i_1\ldots i_l} \psi_0$ because all the terms involving $n_+$ can not contribute.
Indeed, $\Pi n_+\chi=0$ for any $\chi$ because $n_+$ can not map to HW subspace.

\paragraph{Proof of Lemma \bref{lemma:2}.}

 Let us first show that $m_\alpha(\phi)\neq m_{\alpha}(\psi_0)$ iff $\phi$ is trivial
in the sense that $\phi=\smalls^\dagger_\alpha \chi^\alpha$ for some $\chi^\alpha$. To this end
introduce the following notation: $\textsf{n}_{-1}$ denotes an element from the subalgebra
$n_+$ of the form $N_{p+i}^\alpha$, where $p+i$ denote indices running $p+1, ..., n-1$,
$\textsf{n}_0$ either $N_\alpha^\beta$ or $N_{p+i}^{p+j}$ from the subalgebra $n_+$;
$\textsf{s}_0$ denotes $\smalls^\dagger_{p+i}$ and $\textsf{s}_1$ denotes $\smalls^\dagger_\alpha$.
Note that commutation relations have the structure
\begin{equation}
\begin{gathered}
 \commut{\textsf{n}_{-1}}{\textsf{n}_{-1}}=0\,,\quad \commut{\textsf{n}_{-1}}{\textsf{n}_0}=\textsf{n}_{-1}\,,\quad \commut{\textsf{n}_0}{\textsf{n}_0}=\textsf{n}_0\,,\\
\commut{\textsf{n}_0}{\textsf{s}_0}=\textsf{s}_0\,, \quad \commut{\textsf{n}_0}{\textsf{s}_1}=\textsf{s}_1\,,\quad
\commut{\textsf{n}_{-1}}{\textsf{s}_1}=\textsf{s}_0\,, \quad \commut{\textsf{n}_{-1}}{\textsf{s}_0}=0
\end{gathered}
\end{equation}
According to Lemma \bref{lemma:1} a given HW vector can be represented as
$\phi=\Pi\, (\textsf{s}_1)^l(\textsf{s}_0)^{m}\psi_0$ for some nonnegative integers $l,m$. The terms originating from the projector
have the following structure
\begin{equation}
\label{dec}
(\textsf{n}_0)^i (\textsf{n}_{-1})^j (\textsf{s}_1)^{l+j}(\textsf{s}_0)^{m-j}\psi_0
\end{equation}
where the weights $m_y$ and $m_i$ of $\phi$ have been taken into account. Then using the
commutation relations above one moves all $\textsf{s}_1$ to the left.
This results in the expression of the form $\textsf{s}_1(...)$ iff $l>0$. Indeed,
the terms without $\textsf{s}_1$ can arise in this process only if $l=0$
(indeed only commuting $\textsf{n}_{-1}$ with $\textsf{s}_{1}$ one can get rid of
$\textsf{s}_1$; but the power of $\textsf{s}_1$ is higher than that of $\textsf{n}_{-1}$
unless $l=0$). If $l=0$ then analogous arguments show that $\phi$ is nontrivial
$\phi \neq \textsf{s}_1(...)$
and other way around.

\paragraph{Proof of Lemma \bref{lemma:3}.}

It is analogous to that of Lemma \bref{lemma:2}.

In summary, both Lemma \bref{lemma:2} and its generalization Lemma
\bref{lemma:3} mean that nontrivial $ sl(n-1) $ HW elements
representing the equivalence relation  cannot be generated from
the respective $ sl(n) $ HW elements by the first $ p-k $
generators $ \smalls^\dagger_\alpha $ (for $ k=0 $ we recover
Lemma \bref{lemma:2}).

\paragraph{Proof of Lemma \bref{lemma:4}.}

The proof reduces to the following two observations. Firstly, one observes that acting by
$\smalls^\dagger_{i}$ increases a value of weight $s_{i}$
by one and recalls that $sl(n-1)$ HW elements with weights $s_j < s_{j+1}$ vanish identically.
Secondly, given $sl(n-1)$ HW element $\phi$
it is easily seen that the relation $\Pi \smalls^\dagger_{i}\smalls^\dagger_{i+1}\phi = 0$
holds provided that at least two  subsequent weights are equal, \textit{i.e.,} $m_i(\phi) = m_{i+1}(\phi)$.


\providecommand{\href}[2]{#2}\begingroup\raggedright\endgroup
\end{document}